\documentclass{article}
\usepackage{color, xcolor, colortbl}
\usepackage{graphicx}
\usepackage{geometry}
\usepackage{amsthm,amsmath,amssymb,amsfonts}
\usepackage{algorithm}
\usepackage{algorithmic}
\usepackage{dcolumn}
\usepackage{epstopdf}
\usepackage{bm}
\usepackage[caption=false]{subfig}
\usepackage{appendix}
\usepackage{multirow}
\usepackage{braket}
\usepackage[english]{babel}
\usepackage{hyperref}
\usepackage[capitalize]{cleveref}

\usepackage{xpatch}
\usepackage{tikz}
\usepackage{adjustbox}
\usepackage{xspace}
\usepackage{epsfig,graphics,amssymb,amsmath,subeqnarray,mathrsfs,color,xcolor, rotating}

\usepackage{physics}

\renewcommand{\Im}{\operatorname{Im}}

\newcommand{\ud}{\,\mathrm{d}}
\newcommand{\Or}{\mathcal{O}}

\providecommand{\definitionname}{Definition}
\providecommand{\assumptionname}{Assumption}
\providecommand{\corollaryname}{Corollary}
\providecommand{\lemmaname}{Lemma}
\providecommand{\propositionname}{Proposition}
\providecommand{\remarkname}{Remark}
\providecommand{\theoremname}{Theorem}

\newtheorem{thm}{\protect\theoremname}
\newtheorem{lem}[thm]{\protect\lemmaname}

\newtheorem{prop}[thm]{\protect\propositionname}

\newtheorem{assumption}[thm]{\protect\assumptionname}
\newtheorem{defn}[thm]{\protect\definitionname}

\newtheorem{res}[thm]{Result}

\usetikzlibrary{fit}
\tikzset{%
  highlight/.style={rectangle,rounded corners,fill=blue!15,draw,fill opacity=0.3,thick,inner sep=0pt}
}

\usepackage{authblk}

\usepackage[normalem]{ulem}
\usepackage{color, xcolor, colortbl}
\usepackage{graphicx}
\usepackage{geometry}
\usepackage{amsthm,amsmath,amssymb,amsfonts}
\usepackage{algorithm}
\usepackage{algorithmic}
\usepackage{dcolumn}
\usepackage{epstopdf}
\usepackage{bm}
\usepackage[caption=false]{subfig}
\usepackage{appendix}
\usepackage{multirow}
\usepackage{braket}
\usepackage[english]{babel}
\usepackage{hyperref}
\usepackage[capitalize]{cleveref}
\usepackage{xpatch}
\usepackage{tikz}
\usepackage{adjustbox}
\usepackage{xspace}

\begin{document}

\title{Fourier transform-based linear combination of Hamiltonian simulation}
\author[1]{Xi Huang}
\author[2]{Dong An\thanks{dongan@pku.edu.cn}}

\affil[1]{School of Stomatology, Peking University, Beijing, 100081, China}
\affil[2]{Beijing International Center for Mathematical Research, Peking University, Beijing, 100871, China}
\date{\today}
\maketitle

\begin{abstract}
    Linear combination of Hamiltonian simulation (LCHS) connects the general linear non-unitary dynamics with unitary operators and serves as the mathematical backbone of designing near-optimal quantum linear differential equation algorithms. 
    However, the existing LCHS formalism needs to find a kernel function subject to complicated technical conditions on a half complex plane. 
    In this work, we establish an alternative formalism of LCHS based on the Fourier transform. 
    Our new formalism completely removes the technical requirements beyond the real axis, providing a simple and flexible way of constructing LCHS kernel functions. 
    Specifically, we construct a different family of the LCHS kernel function, providing a $1.81$ times reduction in the quantum differential equation algorithms based on LCHS, and an $8.27$ times reduction in its quantum circuit depth at a truncation error of $\epsilon \le 10^{-8}$. 
    Additionally, we extend the scope of the LCHS formula to the scenario of simulating linear unstable dynamics for a short or intermediate time period. 
\end{abstract}


\section{Introduction}

Differential equations are fundamental mathematical tools in modeling a wide range of dynamics, serving as the backbone of mathematical descriptions in physics, chemistry, engineering, economics, and beyond.  
In this work, we consider the ordinary differential equations (ODEs)\footnote{A more general form of linear ODEs is the inhomogeneous ODE $\ud u(t)/\ud t = -A(t)u(t) + b(t)$ for a vector $b(t)$. Notice that as long as we can efficiently simulate its homogeneous version, the general inhomogeneous ODE can be simulated using Duhamel's principle. So in this work we focus on homogeneous ODEs. }
\begin{equation}\label{eqn:def_ODE}
    \frac{\ud u(t)}{\ud t} = -A(t) u(t), \quad u(0) = u_0. 
\end{equation}
Here $A(t) \in \mathbb{C}^{N\times N}$ is the time-dependent coefficient matrix with the Cartesian decomposition~\cite[Chapter I]{Bhatia1997}
\begin{equation}\label{eqn:A_cartesian}
    A(t) = L(t) + iH(t), \quad L(t) = \frac{A(t)+A(t)^{\dagger}}{2}, \quad H(t) = \frac{A(t)-A(t)^{\dagger}}{2i}, 
\end{equation}
and the Hermitian matrices $L(t)$ and $H(t)$ are called the \emph{real} and \emph{imaginary} part of $A(t)$, respectively. 
$u_0 \in \mathbb{C}^{N}$ is the initial condition. 
Linear ODEs usually lie in very high dimensions for accurately modeling complex dynamics in practice. 
Such high dimensionality limits the computational efficiency of classical algorithms and motivates the design of quantum algorithms with the exploration of potential quantum advantage. 

The first quantum algorithm for linear ODEs was proposed by Berry~\cite{Berry2014}. 
The algorithm discretizes~\cref{eqn:def_ODE} by a multi-step method, and solves the resulting linear system of equations by the HHL algorithm~\cite{HarrowHassidimLloyd2009}. 
Since then, there has been significant progress on the design of more efficient quantum algorithms for linear ODEs, either following the same methodology of~\cite{Berry2014} with higher-order time discretization and more advanced quantum linear system algorithms~\cite{BerryChildsOstranderEtAl2017,ChildsLiu2020,Krovi2022,BerryCosta2022,LowSu2024quantumlinearalgorithmoptimal,DongLiXue2025,WuLi2025}, or directly constructing the time-evolution operator of the ODE~\cite{FangLinTong2023,JinLiuYu2022,AnLiuLin2023,AnChildsLin2023,LowSu2024,ShangGuoAnZhao2024,JinLiuMaYu2025,Li2025,FangGeorgeTong2025}. 
Among those, several quantum algorithms~\cite{AnChildsLin2023,LowSu2024,ShangGuoAnZhao2024,JinLiuMaYu2025} have achieved near-optimal complexity, which scales almost linearly in the evolution time $T$ and poly-logarithmically in the error $\epsilon$. 

One of the quantum ODE algorithms with near-optimal complexity is the linear combination of Hamiltonian simulation (LCHS) method~\cite{AnLiuLin2023,AnChildsLin2023}. 
Under the stability condition $L(t) \succeq 0$, the LCHS formalism represents the time evolution operator of~\cref{eqn:def_ODE} as 
\begin{equation}\label{eqn:LCHS_intro}
    \mathcal{T} e^{-\int_0^t A(s) \ud s} = \int_{\mathbb{R}} \frac{f(k)}{ 1-ik} \mathcal{T} e^{-i \int_0^t (kL(s)+H(s)) \ud s} \ud k, 
\end{equation}
where $f(k)$ is a kernel function subject to certain conditions that will be discussed later.  
Intuitively,~\cref{eqn:LCHS_intro} shows that the time evolution operator $\mathcal{T} e^{-\int_0^t A(s) \ud s}$ is a continuous weighted summation of a class of parameterized unitary operator $\mathcal{T} e^{-i \int_0^t (kL(s)+H(s)) \ud s}$, the time-evolution operator governed by the Hamiltonian $kL(s)+H(s)$. 
Therefore, we can construct an efficient quantum ODE algorithm based on a discretized version of~\cref{eqn:LCHS_intro}, by implementing each unitary $\mathcal{T} e^{-i \int_0^t (kL(s)+H(s)) \ud s}$ with efficient Hamiltonian simulation algorithms such as the truncated Dyson series method~\cite{BerryChildsCleveEtAl2015,LowWiebe2019}, and combining them by the quantum linear combination of unitaries (LCU) method~\cite{ChildsWiebe2012,ChildsKothariSomma2017}. 

In the first version of LCHS in~\cite{AnLiuLin2023}, the kernel function is chosen to be $f(k) = \frac{1}{\pi (1+ik)}$, which results in a sub-optimal dependence on the error $\epsilon$. 
This is because the weight $\frac{f(k)}{1-ik}$ decays only quadratically, so we need to truncate the integral in~\cref{eqn:LCHS_intro} at a very large parameter $K = \Or(1/\epsilon)$, and the cost in Hamiltonian simulation increases significantly. 
The subsequent work~\cite{AnChildsLin2023} broadens the scope of LCHS by showing that there is a wider family of kernel functions that can make the LCHS formula in~\cref{eqn:LCHS_intro} hold, and some of them decay much faster. 
More specifically, it gives an explicit choice of the kernel function $f(k) = \frac{1}{2\pi e^{-2^{\beta}} e^{(1+ik)^{\beta}} }$ with $\beta \in (0,1)$ to achieve near-optimal complexity. 

Nevertheless, the technical conditions for the kernel function $f(k)$ established in~\cite{AnChildsLin2023} appear to be complicated and counterintuitive. 
Specifically,~\cite{AnChildsLin2023} requires the kernel function $f(k)$ to be defined not only along the real axis, but also on the entire complex plane (denoted as $f(z)$ for $z \in \mathbb{C}$), and the function $f(z)$ needs to be analytic and decay mildly in the entire lower half plane. 
It seems intuitively peculiar why we need the conditions of $f(z)$ on the complex plane, given the fact that the integral in the LCHS formula~\cref{eqn:LCHS_intro} is only along the real axis. 
Furthermore, such technically complicated conditions pose difficulties in constructing other computable kernel functions to further speed up its decay and improve the efficiency of the LCHS algorithm. 
Additionally, the framework of LCHS in~\cite{AnChildsLin2023} only holds for stable ODEs where the real part $L(t)$ is positive semi-definite. 
This limits the applicable regime of LCHS as many ODEs in practice might be only asymptotically stable or even unstable at least during a sub-interval. 

In this work, we overcome these difficulties by introducing a Fourier transform-based version of LCHS (FT-LCHS). 
We show that as long as the Fourier transform of the function $\frac{f(k)}{1-ik}$ is $e^{-x}$ for $x \geq 0$ with a mild regularity condition on the real axis, the LCHS formula will hold for $L(t) \succeq 0$. 
The FT-LCHS completely removes the requirements of the kernel function beyond the real axis, and provides a simple and flexible way of constructing the kernel function: we can first smoothly extend the function $e^{-x}$ from the positive real axis to the entire real axis, and then take the inverse Fourier transform to obtain $\frac{f(k)}{1-ik}$. 
Using this strategy, we construct a different kernel function $f(k)$ with faster decay. 
This allows us to reduce the total query complexity of the quantum LCHS algorithm by a factor of $1.81$, and reduce the circuit depth of the LCHS algorithm by a factor of $8.27$ at a truncation error of $\epsilon \le 10^{-8}$. 
Additionally, we find that FT-LCHS can be extended to short-time simulation of unstable ODEs where $L(t)$ has negative eigenvalues, if we choose the kernel function from the inverse Fourier transform of $e^{-x}$ with $x \geq -a$ for a sufficiently large $a$. 

\subsection{Main results} 

The main result of this work is the following FT-LCHS formalism. 

\begin{res}[Informal version of~\cref{thm:f_satisfied}]\label{thm:FTLCHS_intro}
    Let $f(k)$ be a function defined on $k \in \mathbb{R}$, and the Fourier transform of $\frac{f(k)}{1-ik}$ be denoted as $F(x)$. 
    If 
    \begin{enumerate}
        \item $F(x) = e^{-x}$ when $x \geq 0$, and 
        \item $F(x) \in C^2(\mathbb{R}) \cap W^{2,1}(\mathbb{R})$\footnote{Here $C^2(\mathbb{R})$ denotes the set of second-order continuously differentiable functions on the real axis, and $W^{2,1}(\mathbb{R})$ denotes the set of functions on the real axis such that the function and its derivatives up to second order are integrable. }. 
    \end{enumerate}
    Then, for any $A(s) = L(s) + i H(s)$ where $L(s) \succeq 0$ for all $0 \leq s \leq t$, we have 
    \begin{equation}
    \mathcal{T} e^{-\int_0^t A(s) \ud s} = \int_{\mathbb{R}} \frac{f(k)}{ 1-ik} \mathcal{T} e^{-i \int_0^t (kL(s)+H(s)) \ud s} \ud k. 
\end{equation}
\end{res}

As discussed earlier, the contribution of~\cref{thm:FTLCHS_intro} is that it completely removes the requirements of the kernel function on the complex plane. 
Furthermore, notice that the first condition of $F(x)$ in~\cref{thm:FTLCHS_intro} is indeed a necessary condition for the LCHS formula to be true. 
This can be observed by setting $A(t) \equiv 1$ and $t = x$ in the LCHS formula. 
Therefore, our~\cref{thm:FTLCHS_intro} also suggests that this necessary condition related to the Fourier transform is almost a sufficient condition as well (subject to a mild regularity condition on the derivatives of the Fourier transform up to the second order). 

However, we remark that~\cref{thm:FTLCHS_intro} does not enlarge the function space of the kernel functions defined in~\cite{AnChildsLin2023}. 
In fact, we prove~\cref{thm:FTLCHS_intro} by showing that the kernel function with conditions in~\cref{thm:FTLCHS_intro} can be analytically extended to the lower half plane and fall into the class in~\cite{AnChildsLin2023}. 
So the function class defined in~\cref{thm:FTLCHS_intro} is just a subset of those in the previous work~\cite{AnChildsLin2023}. 
Nevertheless,~\cref{thm:FTLCHS_intro} yields a practical way to construct the kernel function, whose validity is not obvious from the previous work~\cite{AnChildsLin2023}: 
\begin{enumerate}
    \item Construct a function $F(x)$ by extending $e^{-x}$ from positive real axis to the entire real axis with sufficient smoothness and integrability (i.e., in $C^2(\mathbb{R}) \cap W^{2,1}(\mathbb{R})$). 
    \item Compute the inverse Fourier transform of $F(x)$ to obtain $\frac{f(k)}{1-ik}$. 
\end{enumerate}

Following this approach, we can construct a new kernel function in~\cref{sec:kernel_construction}. 
Although it is not possible to construct a kernel function with asymptotically faster decay than the previous work due to~\cite[Proposition 7]{AnChildsLin2023}, we show that our new kernel function can decay faster than the best previous one up to a constant factor of $8.27$. 
As a result, when we truncate the infinite integral in the LCHS formula to a finite interval $[-K,K]$ for algorithm design, we can choose a much smaller value of the truncation parameter $K$. 
This reduces the quantum circuit depth of the hybrid implementation of the LCHS algorithm by $8.27$ times. 
The total query complexity of the quantum LCHS algorithm depends on the product of the truncation parameter $K$ and the integral value, $\int_{-\infty}^{+\infty} |f(k)/(1-ik)| \ud k$, and our new kernel function can offer a $1.81$ times reduction. 
We summarize the above discussions as follows. 

\begin{res}
    We can construct a new kernel LCHS function such that, at a truncation error of  $\epsilon \le 10^{-8}$, compared with the previous result in~\cite{AnChildsLin2023},  
    \begin{enumerate}
        \item the quantum circuit depth of the LCHS algorithm can be reduced by $8.27$ times, and 
        \item the total quantum query complexity of the LCHS algorithm can be reduced by $1.81$ times. 
    \end{enumerate}
\end{res}

Another application of our FT-LCHS is to extend LCHS to unstable ODEs where the real part $L(t)$ is not positive semi-definite. 
The idea is to force $F(x)$, the Fourier transform of the kernel function $\frac{f(k)}{1-ik}$, to be $e^{-x}$ for a sufficiently large regime of $x$. 
Specifically, in the worst case, the time-evolution operator $\mathcal{T} e^{-\int_0^t A(s) \ud s}$ can be as large as $e^{\lambda_0 t}$, where we assume $(-\lambda_0)$ is a lower bound of the eigenvalues of $L(s)$. 
So it is sufficient to let $F(x) = e^{-x}$ for all $x \geq - \lambda_0 t$, and the LCHS formula can still hold. 

\begin{res}[Informal version of~\cref{thm:fbLCHS-extended}]\label{thm:unstable_intro}
    Let the coefficient matrix to be $A(s) = L(s) + i H(s)$ where $L(s) \succeq -\lambda_0$ for all $0 \leq s \leq T$, where $\lambda_0$ is a non-negative constant and $T$ is the final evolution time. 
    Suppose that the kernel function $f(k)$ satisfies all the conditions in~\cref{thm:FTLCHS_intro} and, additionally, the Fourier transform of $\frac{f(k)}{1-ik}$ is also $e^{-x}$ when $x \geq -\lambda_0 T$. 
    Then, for any $0 \leq t \leq T$, we have 
    \begin{equation}
    \mathcal{T} e^{-\int_0^t A(s) \ud s} = \int_{\mathbb{R}} \frac{f(k)}{ 1-ik} \mathcal{T} e^{-i \int_0^t (kL(s)+H(s)) \ud s} \ud k. 
\end{equation}
\end{res}

Notice that our~\cref{thm:unstable_intro} is valid only on a finite time interval $[0,T]$, unlike the stable case where LCHS holds for any time $t \geq 0$. 
This result cannot be extended to infinite time when $\lambda_0 > 0$, because otherwise $F(x)$ would be $e^{-x}$ for all $x \in \mathbb{R}$ and no longer has an inverse Fourier transform. 

A quantum algorithm for unstable ODEs can be constructed by discretizing the LCHS formula in~\cref{thm:unstable_intro} and then applying the truncated Dyson series method and LCU. 
However, the overall quantum query complexity scales exponentially in the evolution time $T$. 
This is because the query complexity in the LCU step (with amplitude amplification) depends linearly on the $1$-norm of the coefficients, which is roughly $\int_{\mathbb{R}} |f(k)/(1-ik)| \ud k = \Omega(e^{\lambda_0 T})$ (this lower bound can be obtained by taking $A(s) \equiv - \lambda_0$ and $t = T$). 
Therefore, for simulating long-time unstable dynamics, this extended FT-LCHS algorithm is not efficient, which is expected and consistent with the worst-case lower bound in~\cite{AnLiuWangEtAl2023}. 
Nevertheless, our~\cref{thm:unstable_intro} still suggests an efficient method for solving unstable ODEs within a short or intermediate time period when $T = \mathcal{O}(1/\lambda_0)$. 

\subsection{Additional result}

In this work, we also investigate the limitation on the asymptotic decay rate of the LCHS kernel function. 
It has been shown in~\cite[Proposition 7]{AnChildsLin2023} that the LCHS kernel function cannot decay exponentially $\sim e^{-\gamma |k|}$ for a positive number $\gamma$. 
However, similar to the previous LCHS formalism, this lower bound is also established only for analytic functions on the lower half complex plane. 
From there, it remained unclear whether or not there exists a kernel function which is not analytic but satisfies the LCHS formula and decays exponentially. 

In~\cref{sec:constraint}, we show that such a kernel function does not exist either. 
Specifically, we show that as long as the LCHS formula holds, and the function $\frac{f(k)}{1-ik}$ is integrable, then such a kernel function cannot decay exponentially. 
Notice that the integrality of $\frac{f(k)}{1-ik}$ is a necessary condition for an efficient quantum algorithm based on LCU, so our new lower bound completely rules out the possibility of constructing a kernel function with asymptotic decay faster than~\cite{AnChildsLin2023}, even with some unknown approach beyond FT-LCHS.

\subsection{Related works}

After~\cite{AnChildsLin2023}, there have been a few works on further improving or extending the scope of LCHS, such as the LCHS-based quantum eigenvalue transform algorithm~\cite{AnChildsLinYing2024} and the infinite-dimensional LCHS~\cite{LuLiLiuLiu2025}. 
A closely related work~\cite{PocrnicJohnsonKatabarwaWiebe2025} proves  constant factor bounds for LCHS, and provides a $110$ times reduction in runtime cost of the quantum algorithm. 
Compared to~\cite{PocrnicJohnsonKatabarwaWiebe2025}, our technique for improving the constant factor of LCHS is independent and parallel. 
The work~\cite{PocrnicJohnsonKatabarwaWiebe2025} uses the same kernel function as in~\cite{AnChildsLin2023} and the improvements are algorithmic, including a better discretization of the LCHS integral, a more efficient quantum implementation of the LCU step, and a better oblivious amplitude amplification technique, while our work gives a better kernel function. 
Therefore, we could combine our new kernel function with the techniques in~\cite{PocrnicJohnsonKatabarwaWiebe2025} to obtain an even better constant factor in the quantum ODE algorithms based on LCHS.

Besides, we would like to mention the works~\cite{JinLiuMaYu2025,JinLiuMa2024} on improving the Schr\"odingerization method for linear ODEs. 
Schr\"odingerization and LCHS are two closely related frameworks, in the sense that they share similar underlying mathematical principle but construct the algorithms from different perspectives. 
The work~\cite{JinLiuMaYu2025} yields the connection between Schr\"odingerization and Fourier transform, and the work~\cite{JinLiuMa2024} discusses how to extend Schr\"odingerization for unstable dynamics. 
Our work provides a complementary understanding on these issues in the LCHS framework.

\subsection{Discussions and open questions}

In this work, we establish a Fourier transform based LCHS framework, which provides a constant factor improvement in the LCHS algorithm and gives a quantum algorithm for unstable ODEs. 
We also show that as long as the LCHS formula holds, its kernel function cannot decay exponentially. 
This implies that the quantum ODE algorithm based on LCHS cannot achieve optimal scaling (i.e. $\mathcal{O}(\log(1/\epsilon))$) in terms of error $\epsilon$. 
During the final stage of this work, we became aware of a concurrent work~\cite{LowSomma2025}, which establishes a generalized approximate version of LCHS, enabling a kernel function with exponential decay (and thus a quantum ODE algorithm with optimal precision dependence) as well as constant factor improvements. 

An interesting and practically relevant open question is to construct even better kernel functions by combining our FT-LCHS formalism with optimization. 
Specifically, we can try to find an optimizer $f(k)$ in $C^2(\mathbb{R}) \cap W^{2,1}(\mathbb{R})$ subject to the conditions that the Fourier transform of $\frac{f(k)}{1-ik}$ is $e^{-x}$ when $x \geq 0$ and the integral truncation error is bounded by a prefixed $\epsilon$, to minimize the truncation parameter $K$ (for a shorter quantum circuit depth) or the quantity $K \int_{-K}^{K} |f(k)/(1-ik)| \ud k$ (for a lower total quantum query complexity).

\subsection{Organization}

The rest of this paper is organized as follows. 
We first review the existing results of LCHS in~\cref{sec:prelim}. 
Then, we state and prove our FT-LCHS framework in~\cref{sec:fbLCHS_semi}, and show its applications on finding constantly better kernel functions and simulating unstable dynamics in~\cref{sec:app}. 
In~\cref{sec:constraint}, we discuss the constraint of asymptotic decay of the LCHS kernel functions.

\section{Preliminary: existing linear combination of Hamiltonian simulation}\label{sec:prelim}

In this section, we briefly review the existing results of LCHS proposed in~\cite{AnChildsLin2023}. 
The main LCHS formula is given in the following lemma. 

\begin{lem}[{\cite[Theorem 5]{AnChildsLin2023}}]\label{lem:LCHS}
    Let $f(z)$ be a function of $z \in \mathbb{C}$ such that 
    \begin{enumerate} 
        \item (Analyticity) $f(z)$ is analytic on the lower half plane \(\{z: \operatorname{Im}(z) < 0 \}\) and continuous on \(\{z: \operatorname{Im}(z) \leq 0 \}\), 
        \item (Decay) there exists a parameter $\alpha > 0$ such that $|z|^\alpha |f(z)| \leq \widetilde{C}$ for a constant $\widetilde{C}$ when $\Im(z) \leq 0$, and 
        \item (Normalization) $\int_{\mathbb{R}} \frac{f(k)}{ 1-ik} \ud k = 1$. 
    \end{enumerate}
    Let $A(t) \in \mathbb{C}^{N\times N}$ be decomposed according to~\cref{eqn:A_cartesian}. 
    If $L(s)\succeq 0$ for all $0\le s\le t$, then
\begin{equation}\label{eqn:LCHS}
    \mathcal{T} e^{-\int_0^t A(s) \ud s} = \int_{\mathbb{R}} \frac{f(k)}{ 1-ik} \mathcal{T} e^{-i \int_0^t (kL(s)+H(s)) \ud s} \ud k. 
\end{equation}
\end{lem}

Ref.~\cite{AnChildsLin2023} gives several choices of the kernel function with closed-form expression, and the one with the fastest asymptotic decay is 
\begin{equation}\label{eqn:kernel_intro}
    f_{\beta}(z) = \frac{1}{2\pi e^{-2^\beta} e^{(1+iz)^{\beta}} }, \quad \beta \in (0,1).
\end{equation}
When $z = k \in \mathbb{R}$ and $k \rightarrow \infty$, this kernel function decays as $e^{-c |k|^{\beta}\cos(\beta\pi/2)}$. 
Furthermore,~\cite{AnChildsLin2023} proves that~\cref{eqn:kernel_intro} is the near-optimal one in the sense that there does not exist a kernel function which can simultaneously satisfy the conditions in~\cref{lem:LCHS} and decays exponentially as $e^{-c|k|}$. 

\begin{lem}[{\cite[Proposition 7]{AnChildsLin2023}}]\label{lem:LCHS_constraint}
    There does not exist a function $f(z)$ which satisfy assumptions 1-3 in~\cref{lem:LCHS} and 
    \begin{enumerate}
        \item[4.] (Exponential decay) there exist constants $C,\gamma > 0$ such that $|f(k)| \leq C e^{-\gamma|k|}$ for any $z = k \in \mathbb{R}$. 
    \end{enumerate}
\end{lem}

Based on~\cref{lem:LCHS}, we can design a quantum algorithm following the steps: 
\begin{enumerate}
    \item Truncate the integral to a finite interval $[-K,K]$ with a parameter $K > 0$, and discretize the integral by numerical quadrature. 
    This results the approximation 
    \begin{equation}\label{eqn:LCHS_num_approx}
        \mathcal{T} e^{-\int_0^t A(s) \ud s} \approx \int_{-K}^{K} \frac{f(k)}{ 1-ik} \mathcal{T} e^{-i \int_0^t (kL(s)+H(s)) \ud s} \ud k \approx \sum_{j=0}^{M-1} c_j \mathcal{T} e^{-i \int_0^t (k_j L(s)+H(s)) \ud s}. 
    \end{equation}
    Here $M$ is the number of quadrature nodes, $c_j$ and $k_j$ are the weights and nodes, respectively.  
    \item Implement (a controlled version of) the unitary operator $\mathcal{T} e^{-i \int_0^t (k_j L(s)+H(s)) \ud s}$ by the truncated Dyson series method~\cite{LowWiebe2019}. 
    \item Implement the linear combination by LCU. 
\end{enumerate}
Notice that the main cost in the LCHS algorithm is the Hamiltonian simulation step, which takes $\widetilde{\mathcal{O}}(\|A\| K T (\log(1/\epsilon)) )$ query complexity if we use the truncated Dyson series method. 
According to the choice of the kernel function in~\cref{eqn:kernel_intro}, the truncation order $K$ should be chosen as $\mathcal{O}( (\log(1/\epsilon))^{1/\beta} )$. 
Therefore, the overall complexity of the LCHS algorithm is $\widetilde{\mathcal{O}}(\|A\| T (\log(1/\epsilon))^{1+1/\beta} )$. 
Here the decay rate of the kernel function plays a key role in estimating the complexity of the algorithm.

The LCHS formula is closely related to the Fourier transform. 
To see this, let us take the matrix $A(t)$ to be the real number $1$, then we have (with replacing the notation $t$ by $x$) 
\begin{equation}\label{eqn:LCHS_necessary_condition}
    e^{-x} = \int_{\mathbb{R}} \frac{f(k)}{1-ik} e^{-i k x} \ud k, \quad x \geq 0. 
\end{equation}
Let $F(x)$ be a function such that $F(x) = e^{-x}$ for $x \geq 0$, then~\cref{eqn:LCHS_necessary_condition} tells that the Fourier transform of $\frac{f(k)}{1-ik}$ is $F(x)$. 
This further suggests that the Fourier transform of $f(k)$ is $\Phi(x) = F(x) + F'(x)$, which is equal to $0$ when $x \geq 0$. 
Notice that such a condition is only a necessary condition of the LCHS formula from the analysis in~\cite{AnChildsLin2023}.

\section{Fourier transform-based linear combination of Hamiltonian simulation}\label{sec:fbLCHS_semi}

In this section, we present the Fourier transform-based linear combination of Hamiltonian simulation (FT-LCHS) formalism dominated by a set of generating functions that significantly reduces the difficulty of finding the kernel functions while offering a fundamentally more natural proof based on Fourier transform, compared to the original LCHS. 

We first give the set of generating functions and the construction of the kernel function.

\begin{assumption}\label{assump:F}
Let $F(x)$ be a function defined on the real axis. 
We assume $F(x)$ satisfies:
    \begin{enumerate} 
        \item (Behavior on positive axis) $F(x) = e^{-x}$ for all \(x \geq 0\). 
        \item (Regularity) $F(x) \in C^2(\mathbb{R}) \cap W^{2,1}(\mathbb{R})$. 
\end{enumerate}
\end{assumption}

\begin{defn}[The generating function $\Phi(x)$]\label{defn:Phi}
    The generating function $\Phi(x)$ is defined on the real axis as 
    \begin{equation}
        \Phi(x) = F(x) + F'(x). 
    \end{equation}
\end{defn}

\begin{defn}[The kernel function $f(z)$]\label{defn:f}
For any $z \in \{z: \operatorname{Im}(z) \leq 0 \}$, the kernel function $f(z)$ is defined as the inverse Fourier transform of $\Phi(x)$:
\begin{align}
    f(z) = \frac{1}{2\pi} \int_{-\infty}^{\infty} \Phi(x) e^{izx} \, \ud x = \frac{1}{2\pi} \int_{-\infty}^{0} \Phi(x) e^{izx} \, \ud x.
\end{align}
\end{defn}

Notice that for any $z \in \{z: \operatorname{Im}(z) \leq 0 \}$, $f(z)$ through~\cref{defn:f} is well-defined, because $|e^{izx}| \leq 1$ for all $x \leq 0$ and 
\begin{equation}
    \frac{1}{2\pi} \left|\int_{-\infty}^{0} \Phi(x) e^{izx} \, \ud x \right| \leq \frac{1}{2\pi} \int_{-\infty}^{0} |\Phi(x)|  \, \ud x < \infty. 
\end{equation}
Here the last inequality comes from the regularity assumption that $F(x) \in W^{2,1}(\mathbb{R})$ so $\Phi(x) = F(x)+F'(x) \in W^{1,1}(\mathbb{R}) \subset L^1 (\mathbb{R})$. 
Furthermore, although the kernel function $f(z)$ in~\cref{defn:f} is defined on the lower half plane, the LCHS formula~\cref{eqn:LCHS} only involves the values of the kernel functions on the real axis. 
The reason why we still define the kernel function on the complex plane is to facilitate the proof of the LCHS formula: we simply need to verify that such a kernel function satisfies all the conditions in~\cref{lem:LCHS}. 

\begin{thm}\label{thm:f_satisfied}
Under Assumption \ref{assump:F}, $f(z)$ satisfies the conditions in~\cref{lem:LCHS}, and thus for any $A(t) = L(t) + iH(t)$ with $L(s) \succeq 0$ for $0 \leq s \leq t$, we have 
\begin{equation}
    \mathcal{T} e^{-\int_0^t A(s) \ud s} = \int_{\mathbb{R}} \frac{f(k)}{ 1-ik} \mathcal{T} e^{-i \int_0^t (kL(s)+H(s)) \ud s} \ud k. 
\end{equation}
\end{thm}
\begin{proof}
By~\cref{assump:F} and the definition of the generating function $\Phi(x)$ in~\cref{defn:Phi}, we have that 
\begin{equation}
    \Phi(x) = 0, \, \forall \, x \geq 0, 
\end{equation}
and 
\begin{equation}
    \Phi(x) \in C^1(\mathbb{R}) \cap W^{1,1}(\mathbb{R}). 
\end{equation}
Now we verify the three conditions in~\cref{lem:LCHS} for our kernel function $f(z)$ generated by $\Phi(x)$ through~\cref{defn:f}. 
\begin{enumerate}
    \item[\textnormal{(i)}] \textbf{Analyticity:} Since $\Phi(x)$ is identically zero for all \(x > 0\), and for any $z=k+iw$ with $\operatorname{Im}(z)=w \leq 0$, the exponential term $|e^{izx}|=|e^{-wx}| \leq 1$ if $x \leq 0$, we have that $f(z)$ converges absolutely and uniformly: 
\begin{align}
      |f(z)| \le \frac{1}{2\pi} \int_{-\infty}^{0} |\Phi(x)| |e^{izx}| \, \ud x \le \frac{1}{2\pi} \int_{-\infty}^{0} |\Phi(x)| \, \ud x = \frac{\|\Phi\|_{L^1}}{2\pi} < \infty.
\end{align}
This shows \( f(z) \) is well-defined on \(\{z: \operatorname{Im}(z) \leq 0 \}\). 

Let $z_0$ be an arbitrary point in \(\{z: \operatorname{Im}(z) \leq 0 \}\), and $\{z_n\}_{n=1}^\infty$ be any sequence in \(\{z: \operatorname{Im}(z) \leq 0 \}\) such that $\lim_{n \to \infty} z_n = z_0$. 
We then apply the dominated convergence theorem to interchange the limit and the integral and obtain
\begin{align}
    \lim_{n \to \infty} f(z_n) &= \lim_{n \to \infty} \frac{1}{2\pi} \int_{-\infty}^{0} \Phi(x) e^{iz_n x} \, \ud x \\
    &= \frac{1}{2\pi} \int_{-\infty}^{0} \left( \lim_{n \to \infty} \Phi(x) e^{iz_n x} \right) \, \ud x \\
    &= \frac{1}{2\pi} \int_{-\infty}^{0} \Phi(x) e^{iz_0 x} \, \ud x \\
    &= f(z_0).
\end{align}
This shows that $f(z)$ is continuous on \(\{z: \operatorname{Im}(z) \leq 0 \}\). 

Let \( \Gamma \subset \{z: \operatorname{Im}(z) < 0 \} \) be any closed piecewise $C^1$ contour. 
The function $\Phi(x) e^{izx}$ is still absolutely integrable over $(z,x) \in \Gamma \times (-\infty,0)$. 
By the Fubini-Tonelli theorem, we obtain
  \begin{align}
        \oint_{\Gamma} f(z) \ud z &= \oint_{\Gamma} \left( \frac{1}{2\pi} \int_{-\infty}^{0} \Phi(x) e^{izx} \, \ud x \right) \ud z   = \frac{1}{2\pi} \int_{-\infty}^{0} \Phi(x) \left( \oint_{\Gamma} e^{izx} \, \ud z \right) \ud x.
    \end{align}
Cauchy's integral theorem gives that \(\oint_{\Gamma} e^{-izx} \, dz = 0 \) for any $x$. 
Hence
\begin{align}
   \oint_{\Gamma} f(z) \ud z=0. 
\end{align}
By Morera's theorem, $f(z)$ is analytic on the lower half plane \(\{z: \operatorname{Im}(z) < 0 \}\). 

    \item[\textnormal{(ii)}] \textbf{Decay:} We perform integration by parts on the definition of $f(z)$ and obtain 
    \begin{align}
        f(z) & = \frac{1}{2\pi} \int_{-\infty}^{0} \Phi(x) e^{izx} \, \ud x \\
        & = \frac{1}{2\pi i z} \left( \Phi(0) - \lim_{x\rightarrow -\infty} ( \Phi(x) e^{izx} ) \right) - \frac{1}{2\pi i z} \int_{-\infty}^{0} \Phi'(x) e^{izx} \, \ud x \\
        & = - \frac{1}{2\pi i z} \lim_{x\rightarrow -\infty} ( \Phi(x) e^{izx} ) - \frac{1}{2\pi i z} \int_{-\infty}^{0} \Phi'(x) e^{izx} \, \ud x. 
    \end{align}
    Notice that $\lim_{x\rightarrow -\infty} \Phi(x)$ exists due to the fundamental theorem of calculus $\Phi(x) = \Phi(0) - \int_{x}^0 \Phi'(y) \ud y$ and the existence of $ \lim_{x\rightarrow -\infty} \int_{x}^0 \Phi'(y) \ud y$ for $\Phi'(x) \in L^1(\mathbb{R})$. 
    Then, $\lim_{x\rightarrow -\infty} \Phi(x)$ must be $0$ since $\Phi(x) \in L^1(\mathbb{R})$. 
    Combining with $|e^{izx}| \leq 1$ for $\operatorname{Im}(z) \leq 0$ and $x \leq 0$, we have $\lim_{x\rightarrow -\infty} ( \Phi(x) e^{izx} ) = 0$ and 
    \begin{equation}
        f(z) = - \frac{1}{2\pi i z} \int_{-\infty}^{0} \Phi'(x) e^{izx} \, \ud x. 
    \end{equation}
    Then 
    \begin{align}
        |f(z)| \le \frac{1}{2\pi |z|} \int_{-\infty}^{0} |\Phi'(x)| \ud x = \frac{\|\Phi'\|_{L^1}}{2\pi|z|}, 
    \end{align}
    which suggests that the decay condition in~\cref{lem:LCHS} holds with $\alpha = 1$ and $\widetilde{C} = \|\Phi'\|_{L^1}/(2\pi)$. 
    
    \item[\textnormal{(iii)}] \textbf{Normalization:} Notice that for any $k \in \mathbb{R}$, we have 
    \begin{align}
        \frac{f(k)}{1-ik} &= \frac{1}{2\pi(1-ik)} \int_{-\infty}^{+\infty} \Phi(x) e^{ikx} \, \ud x \\
        &= \frac{1}{2\pi(1-ik)} \left( \int_{-\infty}^{+\infty} F(x) e^{ikx} \, \ud x + \int_{-\infty}^{+\infty} F'(x) e^{ikx} \, \ud x \right) \\
        & = \frac{1}{2\pi (1-ik)} \left( \int_{-\infty}^{+\infty} F(x) e^{ikx} \, \ud x  - ik \int_{-\infty}^{+\infty} F(x) e^{ikx} \, \ud x \right) \\
        & = \frac{1}{2\pi} \int_{-\infty}^{+\infty} F(x) e^{ikx} \, \ud x. 
    \end{align}
    So the Fourier transform of $\frac{f(k)}{1-ik}$ is $F(x)$, and thus $\int_{\mathbb{R}} \frac{f(k)}{ 1-ik} \ud k = F(0) = 1$ from~\cref{assump:F}.
\end{enumerate}
\end{proof}

\section{Applications of FT-LCHS}\label{sec:app}

In this section, we discuss how our FT-LCHS formalism can improve and extend the scope of the LCHS algorithm. 
We first show a concrete approach to construct a kernel function different from the choice in previous works, reducing the quantum query complexity of the LCHS algorithm by up to a multiplicative factor of $1.81$, and the quantum circuit depth of the LCHS algorithm by up to a multiplicative factor of $8.27$. 
We then show how the LCHS formalism can be extended to the case of unstable dynamics, where the real part of the coefficient matrix $A(t)$, $L(t)$, is not necessarily positive semi-definite.

\subsection{Constantly improved LCHS algorithm}\label{app:applications_proof}

The FT-LCHS framework allows flexibility in kernel function design. 
We do not need to worry about the values (or even the analyticity and decay) of the kernel function over the entire lower half complex plane. 
Instead, we can simply follow the steps: 
\begin{enumerate}
    \item Construct a function $F(x)$ by extending $e^{-x}$ from positive real axis to the entire real axis with sufficient smoothness and integrability such that~\cref{assump:F} is satisfied. 
    \item Compute the kernel function $f(k)$ by taking the inverse Fourier transform of $F(x)$ and then multiplying $(1-ik)$.  
\end{enumerate}
Although this is an effective and constructive approach, the resulting kernel functions are not expected to decay asymptotically faster than those in the original framework in~\cite{AnChildsLin2023}. 
This is because the kernel functions constructed from FT-LCHS satisfy all the conditions in~\cref{lem:LCHS} as shown in~\cref{thm:f_satisfied}, so the no-go theorem in~\cref{lem:LCHS_constraint} still applies. 
Nevertheless, it is possible to construct kernel functions with constantly faster decay at least in a sub-regime, thus improving the efficiency of the quantum LCHS algorithm up to a constant factor.

\subsubsection{Criteria of a good kernel function}

Let us first discuss how to determine a good kernel function, i.e., the properties of a kernel function that will affect the efficiency of the LCHS algorithms. 
As discussed in~\cref{sec:prelim} after~\cref{lem:LCHS}, we first truncate the integral in the LCHS formula to a finite interval $[-K,K]$ and discretize it by numerical quadrature, leading to the approximation in~\cref{eqn:LCHS_num_approx}. 
Then, in a quantum implementation of the LCHS algorithm, we implement the Hamiltonian simulation problem by the truncated Dyson series method, and linearly combine them by LCU. 

The overall complexity of the algorithm is dominated by the number of queries to the coefficient matrix $A(t)$, which is affected by two components. 
Here we only focus on the asymptotic scaling in terms of the parameters related to the kernel function, so we assume the evolution time and the accuracy are fixed constants. 
First, in each run of LCU, we need to construct the select oracle of (i.e., a coherently controlled version of) the Hamiltonian simulation operators $\mathcal{T} e^{-i \int_0^t ( k_j L(s) + H(s)) \ud s } $. 
Using truncated Dyson series method, the complexity scales linearly in the spectral norm of the Hamiltonians, which is upper bounded by $\mathcal{O}(K)$. 
Here $K$ is chosen such that the integral truncation error is bounded by a pre-fixed error parameter. 
Second, the LCU needs to be repeated for multiple times to achieve high success probability. 
With amplitude amplification, the number of rounds is $\mathcal{O}(\|\vec{c}\|_1)$, where $\|\vec{c}\|_1 = \sum_{j=0}^{M-1} |c_j|$. 
Therefore, the overall query complexity is 
\begin{equation}
    \mathcal{O}(\|\vec{c}\|_1 K). 
\end{equation}
According to the definition of the coefficients $c_j$ in~\cref{eqn:LCHS_num_approx}, we have $\|\vec{c}\|_1 \approx \int_{-K}^K \left|\frac{f(k)}{1-ik}\right| \ud k $, so the overall complexity becomes 
\begin{equation}
    \mathcal{O}\left(K \int_{-K}^K \left|\frac{f(k)}{1-ik}\right| \ud k \right), 
\end{equation}
which can be used as a measure of the quality of a kernel function. 

In some cases, the truncation parameter $K$ itself can also serve as an indicator of the quality of a kernel function. 
For example, if we wish to reduce the circuit depth of the quantum algorithm, we can independently run LCU circuits in parallel and measure the ancilla qubits until success without amplitude amplification. 
Here the circuit depth is simply $\mathcal{O}(K)$. 
Another scenario is when we would like to estimate the unnormalized observable $u(t)^{\dagger} O u(t)$. 
Here $u(t)$ is the ODE solution, and $O$ is a Hermitian matrix. 
Using the approximate LCHS formula in~\cref{eqn:LCHS_num_approx}, we can write the observable as 
\begin{equation}
    u(t)^{\dagger} O u(t) \approx \sum_{j=0}^{M-1} \sum_{j'=0}^{M-1} \overline{c}_{j} c_{j'} \langle u_0 | U_j^{\dagger} O U_{j'} | u_0 \rangle, 
\end{equation}
where $U_j = \mathcal{T} e^{-i \int_0^t (k_j L(s)+H(s)) \ud s}$. 
So we can estimate each observable $\langle u_0 | U_j^{\dagger} O U_{j'} | u_0 \rangle$ independently by Hadamard test, and then linearly combine them on a classical computer by importance sampling. 
The quantum circuit depth is again $\mathcal{O}(K)$. 

In summary, we can use the following two quantities to measure the quality of a kernel function: 
\begin{enumerate}
    \item The truncation parameter $K$, which is the quantum circuit depth of the LCHS algorithm without amplitude amplification or the hybrid implementation of LCHS. 
    \item The quantity $K \int_{-K}^K \left|\frac{f(k)}{1-ik}\right| \ud k$, which is the overall complexity of the quantum LCHS algorithm with amplitude amplification. 
\end{enumerate}

\subsubsection{A specific construction}\label{sec:kernel_construction}
Although~\cref{assump:F} only requires $F(x)$ to be second-order differentiable, in practice we usually choose $F(x)$ in the Schwartz space to ensure that the corresponding kernel function $f(k)$ decays at least super-polynomially. 
To this end, we smoothly connect a zero function and $e^{-x}$. 
Specifically, for a non-negative number $\delta \geq 0$, we define 
\begin{align}\label{eqn:def_kernel_from_glue}
F(x) =
\begin{cases}
e^{-x}, & x > 0, \\
\psi \left(\frac{x + \delta}{\delta}\right) e^{-x}, & -\delta \leq x \leq 0, \\
0, & x < -\delta,
\end{cases}
\end{align}
where the transition from $0$ to $e^{-x}$ occurs in the interval $[-\delta, 0]$. 
The function $\psi(x)$ is called the glue function, defined for $y \in [0,1]$ as
\begin{equation}
\psi(y) = c_g \int_{0}^{y} e^{-\frac{1}{p^{m}(1-p)^{m}}} \ud p, 
\end{equation}
with \(c_g = \left(\int_{0}^{1} e^{-\frac{1}{p^{m}(1-p)^{m}}} \ud p \right)^{-1}\) as the normalization constant and $m$ being a positive integer. 
Notice that $\psi(y)$ is a $C^{\infty}$ function with $\psi(0) = 0$, $\psi(1) = 1$, and vanishing derivatives of any order at $y = 0$ and $1$. 
As a result, $F(x)$ in~\cref{eqn:def_kernel_from_glue} is $C^{\infty}$ and its derivative of any order is integrable. 
We use $g(k)$ to denote the inverse Fourier transform of $F(x)$ in~\cref{eqn:def_kernel_from_glue}, which is supposed to be $\frac{f(k)}{1-ik}$ for some kernel function $f(k)$. 

\begin{figure}
    \centering
    \includegraphics[width=0.75\textwidth]{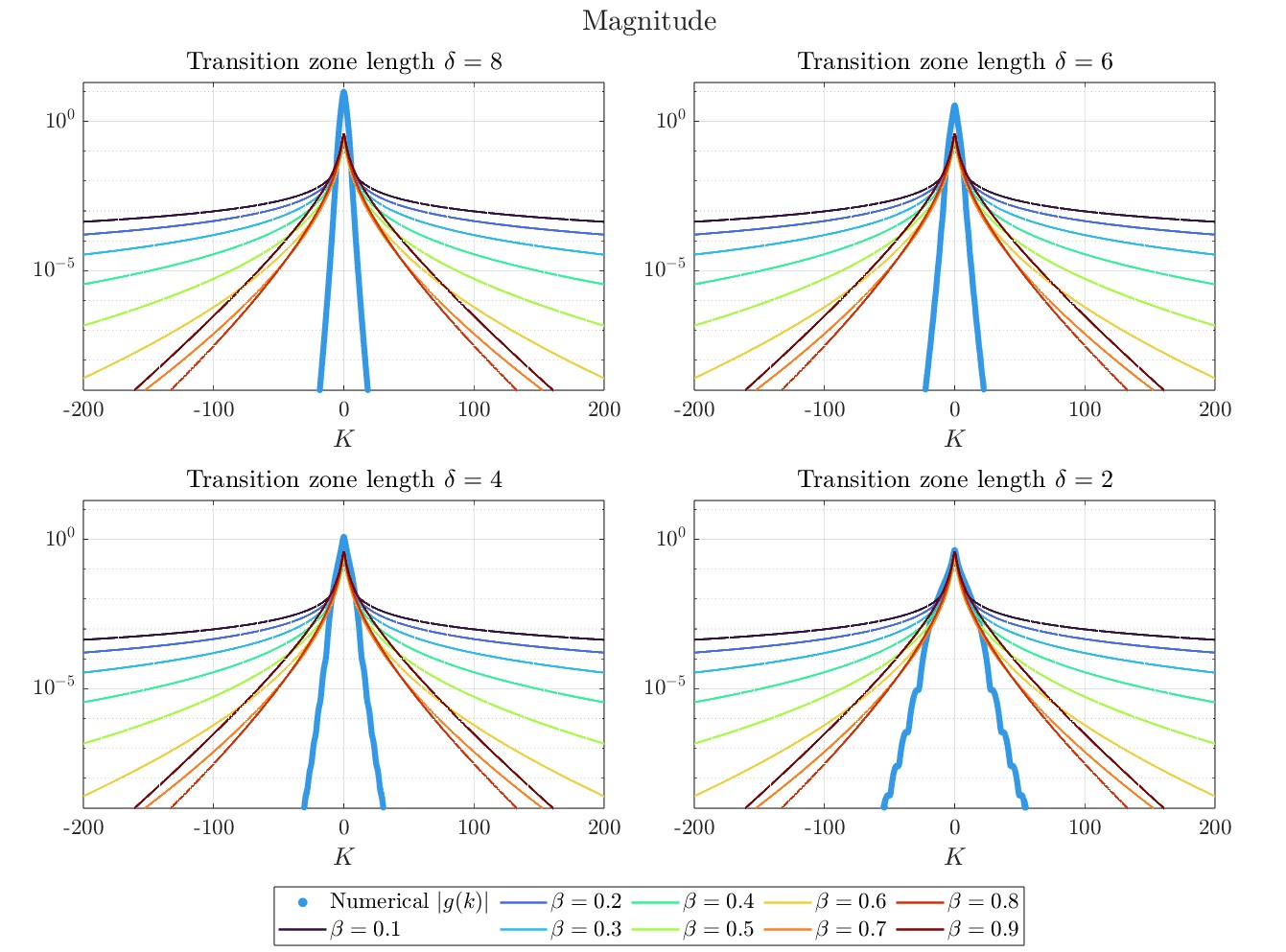}
    \caption{ The magnitude of the \(g(k)\) (with \( m=2,\text{transition zone length}=2,4,6,8\)) compared to \(|f_\beta(k)/(1-ik)|\) from the original LCHS for various $\beta \in [0.1, 0.9]$. }
    \label{fig:magnitude}
\end{figure}

\cref{fig:magnitude} shows the magnitude of $|g(k)|$ with $m = 2$ and various values of $\delta$ compared to $|f_\beta(k)/(1-ik)|$. 
We observe that our kernel function decays much more quickly compared to the one used in the previous work~\cite{AnChildsLin2023}. 
Furthermore, our kernel function decays faster as the transition length $\delta$ increases, suggesting a smaller truncation parameter $K$. 

We then compare our kernel function against the original one ($\beta=0.7, 0.8$) in a practical setting. 
We numerically evaluate the minimum $K$ to achieve a target truncation error $\epsilon$ for randomly generated $4 \times 4$ matrices $A$. 
The results are shown in~\cref{fig:min_k_error} and suggest that the new kernel functions are more efficient compared to the original one. 
Specifically, a larger transition zone length $\delta$ leads to a smaller required $K$. 
Besides, the kernel function with order $m=2$ appears to offer the best numerical performance in the requirement for \(K\). 
At a truncation error of $\epsilon \le 10^{-8}$, our kernel function can reduce the required $K$ from the original LCHS ($\beta=0.7, 0.8$) by up to a factor of nearly $9.71$ and $8.27$. 

Nevertheless, \cref{fig:magnitude} also shows that the magnitude of our kernel functions near $0$ can be greater than the original one, and even greater as $\delta$ increases, leading to a larger value of the integral $\int_{-K}^{K}|g(k)|\ud k$. 
To study whether our kernel function could offer a complexity saving in the quantum LCHS algorithm, we compare the value $K\cdot\int_{-K}^{K}|f(k)/(1-ik)|\ud k$ of our kernel function with various choices of $\delta$ and $m$ against the original kernel function. 
The results are shown in~\cref{fig:complexity}. 
In this case, the kernel function with carefully tuned $m$ and $\delta$ appears to offer a better balance between the truncation parameter and the integral value and thus yields smaller $K\cdot\int_{-K}^{K}|f(k)/(1-ik)|\ud k$. 
At a truncation error of $\epsilon \le 10^{-8}$, our kernel function with $m = 1$ and $\delta = 2$ can reduce the value of $K\cdot\int_{-K}^{K}|f(k)/(1-ik)|\ud k$ by a factor of $1.93$ and $1.81$, nearly halving the overall complexity of the quantum LCHS algorithm ($\beta=0.7, 0.8$).

\begin{figure}
    \centering
    \includegraphics[width=1.0\textwidth]{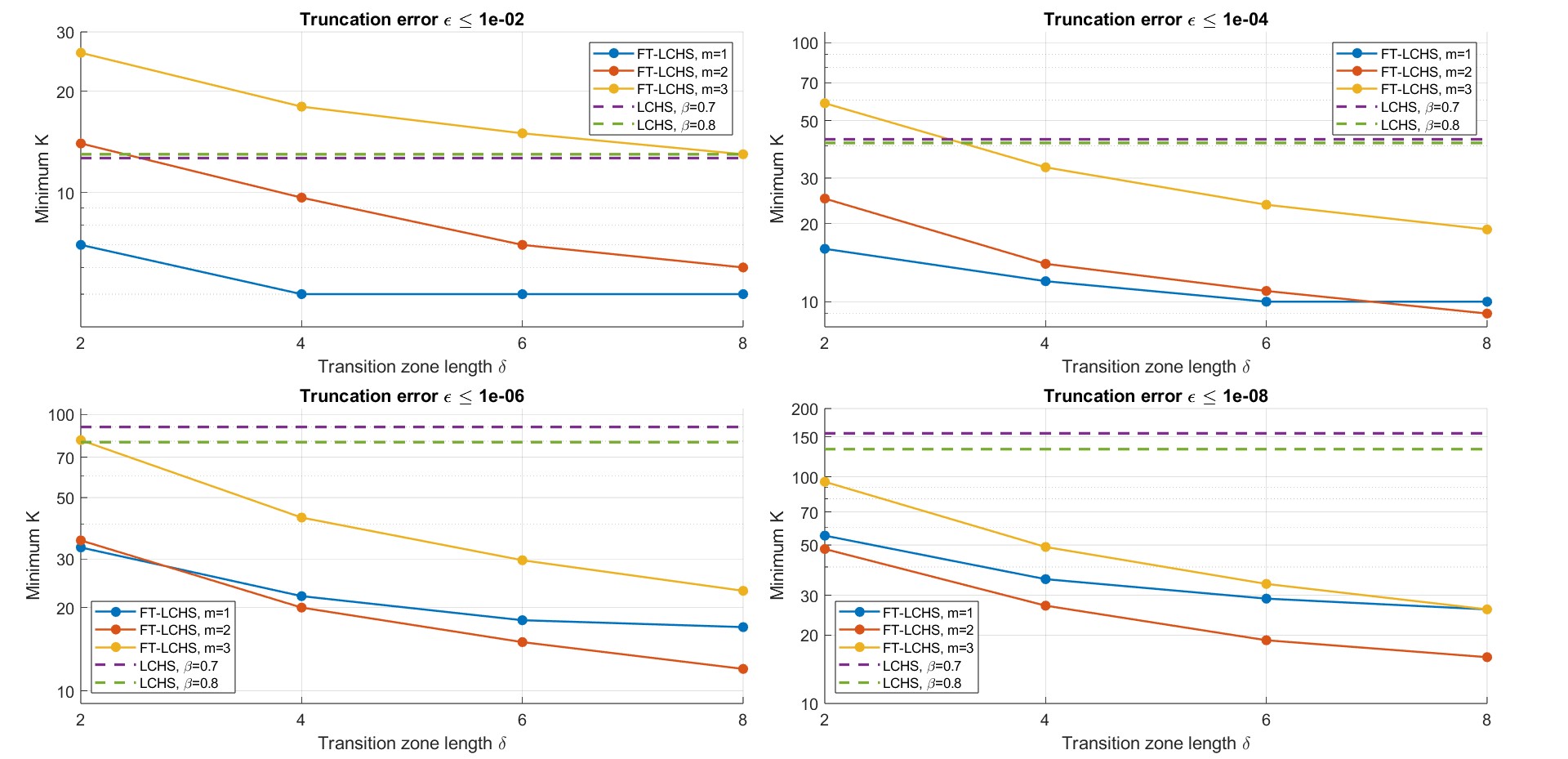}
    \caption{ Minimum $K$ required to achieve a series of truncation errors, $\epsilon$, between FT-LCHS (with order $m=1,2,3$) and LCHS (with $\beta=0.7,0.8$) for varying transition zone lengths $\delta$.}
    \label{fig:min_k_error}
\end{figure}

\begin{figure}
    \centering
    \includegraphics[width=1.0\textwidth]{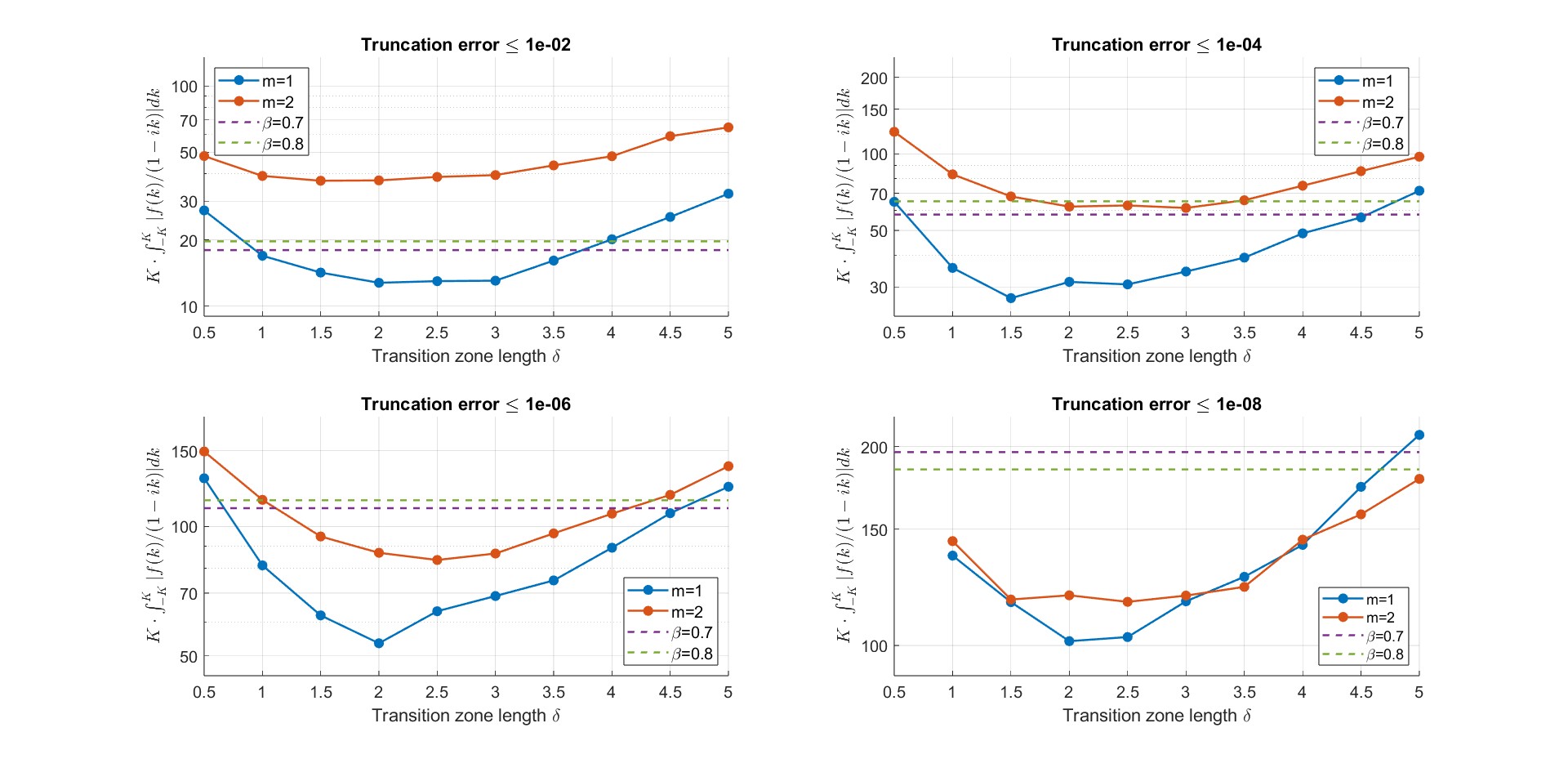}
    \caption{\(K\cdot\int_{-K}^{K}|f(k)/(1-ik)|\ud k\) between FT-LCHS (with order $m=1,2$) and LCHS (with $\beta=0.7,0.8$) for varying transition zone lengths $\delta$.}
    \label{fig:complexity}
\end{figure}

\subsection{LCHS for unstable dynamics}\label{sec:fbLCHS_away_origin}

The original LCHS, along with many other methods for linear ODEs, assumes the real part of \( A(t) \) to be positive semi-definite. 
In other words, the log-norm of $A(t)$ is non-negative. 
However, by employing the FT-LCHS with support of the generating function including part of the negative real axis, we can naturally extend the applicability of LCHS to cases where the log-norm of $A(t)$ is negative. 
However, this extension is rigorously valid only within a certain time limit, which is inversely proportional to the magnitude of the negative log-norm. 

We first state the assumption for the function $F(x)$ in this subsection. 

\begin{assumption}\label{assump:strong}
Let $F(x)$ be a function defined on the real axis. 
We assume $F(x)$ satisfies:
    \begin{enumerate} 
        \item (Behavior on real axis) $F(x) = e^{-x}$ for all \(x \geq -a\) where $a$ is a non-negative number. 
        \item (Regularity) $F(x) \in C^2(\mathbb{R}) \cap W^{2,1}(\mathbb{R})$. 
\end{enumerate}
\end{assumption}

Compared to~\cref{assump:F} in the stable case, the above~\cref{assump:strong} is forced to be $e^{-x}$ in an extended regime $[-a, \infty)$. 
We can construct the corresponding kernel function following the same approach in~\cref{sec:fbLCHS_semi}: we first compute the generating function $\Phi(x)$ by~\cref{defn:Phi} and take the inverse Fourier transform to obtain $f(k)$ as in~\cref{defn:f}. 
Then we have the following extended LCHS formalism. 

\begin{thm} \label{thm:fbLCHS-extended}
Let $A(t) = L(t) + i H(t) \in \mathbb{C}^{N\times N}$ be the coefficient matrix such that $L(t) \succeq -\lambda_0$ for a non-negative real number $\lambda_0$. 
Let $F(x)$ be a function satisfying~\cref{assump:strong} with parameter $a$, and $f(k)$ be a kernel function obtained according to~\cref{defn:Phi} and~\cref{defn:f}. 
Then, for all $t \leq a/\lambda_0$, we have  
\begin{equation}
    \mathcal{T} e^{-\int_0^t A(s) \ud s} = \int_{\mathbb{R}} \frac{f(k)}{ 1-ik} \mathcal{T} e^{-i \int_0^t (kL(s)+H(s)) \ud s} \ud k. 
\end{equation}
\end{thm} 
\begin{proof}
    We will prove this claim by reducing it to the stable case via proper shifting. 
    For $0 \leq s \leq t$, we define
    \begin{equation}
        \widetilde{L}(s) = L(s) + \frac{a}{t} I \succeq 0. 
    \end{equation}
    Then 
    \begin{equation}\label{eqn:proof_extended_eq1}
        \int_{\mathbb{R}} \frac{f(k)}{ 1-ik} \mathcal{T} e^{-i \int_0^t (kL(s)+H(s)) \ud s} \ud k = \int_{\mathbb{R}} \frac{f(k) e^{iak} }{ 1-ik} \mathcal{T} e^{-i \int_0^t (k\widetilde{L}(s)+H(s)) \ud s} \ud k. 
    \end{equation}
    Notice that the function $\frac{f(k)}{1-ik}$ is the inverse Fourier transform of $F(x)$, so $e^{-a}\frac{f(k)e^{iak}}{1-ik}$ is the inverse Fourier transform of $e^{-a} F(x-a)$. 
    Since the function $e^{-a} F(x-a)$ satisfies~\cref{assump:F} and $\widetilde{L}(s) \succeq 0$, we can apply~\cref{thm:f_satisfied} to the matrix $\widetilde{L}(s) + iH(s)$ and obtain 
    \begin{equation}\label{eqn:proof_extended_eq2}
        \mathcal{T} e^{-\int_0^t (\widetilde{L}(s) + i H(s)) \ud s} = \int_{\mathbb{R}} e^{-a}\frac{f(k) e^{iak}}{ 1-ik} \mathcal{T} e^{-i \int_0^t (k\widetilde{L}(s)+H(s)) \ud s} \ud k. 
    \end{equation}
    Plugging~\cref{eqn:proof_extended_eq2} back to~\cref{eqn:proof_extended_eq1}, we have 
    \begin{equation}
        \int_{\mathbb{R}} \frac{f(k)}{ 1-ik} \mathcal{T} e^{-i \int_0^t (kL(s)+H(s)) \ud s} \ud k = e^a \mathcal{T} e^{-\int_0^t (\widetilde{L}(s) + i H(s)) \ud s} = \mathcal{T} e^{-\int_0^t (L(s) + i H(s)) \ud s}. 
    \end{equation}
\end{proof}

\section{Constraint of kernel function}\label{sec:constraint}

As discussed earlier, although~\cref{thm:f_satisfied} provides an effective way to construct kernel functions, the decay properties of the kernel functions remain fundamentally the same as in~\cref{lem:LCHS_constraint}, because the kernel functions in FT-LCHS lie in a subspace of the kernel functions specified in~\cref{lem:LCHS}. 
In this section, we will show a slightly stronger result that the kernel function cannot decay exponentially as long as the LCHS formula of the form~\cref{eqn:LCHS} holds and the weight function $\frac{f(k)}{1-ik}$ is $L^1(\mathbb{R})$, which appears to be weaker than the regularity assumption made in~\cref{assump:F} and is the necessary condition for an efficient quantum algorithm implemented by LCU.

\begin{prop}\label{prop:constraint}
    There does not exist a function $f(z)$ of $z \in \mathbb{C}$ such that 
    \begin{enumerate}
        \item (LCHS formula) For any $A(t) = L(t) + iH(t) \in \mathbb{C}^{N \times N}$ with $L(s) \succeq 0$ for $0 \leq s \leq t$, it holds
        \begin{equation}
        \mathcal{T} e^{-\int_0^t A(s) \ud s} = \int_{\mathbb{R}} \frac{f(k)}{ 1-ik} \mathcal{T} e^{-i \int_0^t (kL(s)+H(s)) \ud s} \ud k, 
        \end{equation}
        \item (Integrability) $\frac{f(k)}{1-ik} \in L^1(\mathbb{R})$, and 
        \item (Decay) $ |f(k)| \leq Ce^{-\gamma|k|}$ for constants $C>0$ and \(\gamma>0\).  
    \end{enumerate}
\end{prop}
\begin{proof}
We will prove the claim by contradiction. 
Suppose that there exists such a function $f(z)$. 
Consider \( z = x + iy \) with \( |y| = |\operatorname{Im}(z)| < \gamma \). 
Let 
\begin{equation}
    g(k) = \frac{f(k)}{1-ik}, 
\end{equation}
and we define 
\begin{align}
    \Psi(z) = \int_{\mathbb{R}}g(k) e^{-ikz} \ud k. \label{eqn:analy_continue}
\end{align}
We bound the integrand of~\cref{eqn:analy_continue}  as
\begin{align}
    |g(k) e^{-ikz}| = |g(k)| e^{ky} \leq C e^{-\gamma|k|} e^{|k| |y|} = C e^{-|k|(\gamma - |y|)}. 
\end{align}
Then
\begin{align}\label{eqn:proof_constraint_eq1}
    |\Psi(z)|\leq C\int_{\mathbb{R}} e^{-|k|(\gamma - |y|)} \ud k = 2C \int_{0}^{\infty} e^{-k(\gamma - |y|)} \ud k = \frac{2C}{\gamma - |y|} < \infty,
\end{align}
so the function $\Psi(z)$ is well defined for $z \in \left\{ |\operatorname{Im}(z)| < \gamma \right\}$. 

Let \( \Gamma \) be any closed piecewise $C^1$ contour in $\left\{ |\operatorname{Im}(z)| < \gamma \right\}$. 
By the Fubini-Tonelli theorem, we have 
\begin{align}
    \oint_{\Gamma} \Psi(z) \ud z = \oint_{\Gamma} \left( \int_{\mathbb{R}} g(k) e^{-ikz} \ud k \right) \ud z = \int_{\mathbb{R}} g(k) \left( \oint_{\Gamma} e^{-ikz} \ud z \right) \ud k. 
\end{align}
Cauchy's integral theorem gives that \(\oint_{\Gamma} e^{-ikz} \, dz = 0 \) for all  \(k \in \mathbb{R}\), so
\begin{align}
   \oint_{\Gamma} \Psi(z) \ud z=0. 
\end{align}
Thus \( \Psi(z) \) is analytic in $\left\{ |\operatorname{Im}(z)| < \gamma \right\}$ by Morera's theorem. 

By the LCHS formula with dimension $N = 1$ and $A(t) \equiv 1$, we have that for all $x \geq 0$ 
\begin{equation}
    \Psi(x) = \int_{\mathbb{R}} g(k) e^{-i k x } \ud k = e^{-x},  
\end{equation}
which implies that $\Psi(z) = e^{-z}$ for all $z \in \left\{ |\operatorname{Im}(z)| < \gamma \right\}$ according to the identity theorem. 
In particular, $\Psi(x) = e^{-x}$ for all $x \in \mathbb{R}$. 
This contradicts the boundedness of $\Psi(x)$ in~\cref{eqn:proof_constraint_eq1} and completes the proof. 

\end{proof}

\section*{Acknowledgments}

DA acknowledges the support by the Innovation Program for Quantum Science and Technology via Project 2024ZD0301900, and the Fundamental Research Funds for the Central Universities, Peking University.

\bibliographystyle{unsrt}
\bibliography{ref}

\begin{thebibliography}{10}

\bibitem{Bhatia1997}
Rajendra Bhatia.
\newblock {\em Matrix Analysis}, volume 169 of {\em Graduate Texts in Mathematics}.
\newblock Springer, 1997.

\bibitem{Berry2014}
Dominic~W. Berry.
\newblock High-order quantum algorithm for solving linear differential equations.
\newblock {\em Journal of Physics A: Mathematical and Theoretical}, 47(10):105301, 2014.

\bibitem{HarrowHassidimLloyd2009}
Aram~W. Harrow, Avinatan Hassidim, and Seth Lloyd.
\newblock Quantum algorithm for linear systems of equations.
\newblock {\em Physical Review Letters}, 103(15):150502, October 2009.

\bibitem{BerryChildsOstranderEtAl2017}
Dominic~W. Berry, Andrew~M. Childs, Aaron Ostrander, and Guoming Wang.
\newblock Quantum algorithm for linear differential equations with exponentially improved dependence on precision.
\newblock {\em Communications in Mathematical Physics}, 356(3):1057--1081, 2017.

\bibitem{ChildsLiu2020}
Andrew~M. Childs and Jin-Peng Liu.
\newblock Quantum spectral methods for differential equations.
\newblock {\em Communications in Mathematical Physics}, 375(2):1427–1457, February 2020.

\bibitem{Krovi2022}
Hari Krovi.
\newblock Improved quantum algorithms for linear and nonlinear differential equations.
\newblock {\em Quantum}, 7:913, 2 2023.

\bibitem{BerryCosta2022}
Dominic~W. Berry and Pedro C.~S.~Costa.
\newblock Quantum algorithm for time-dependent differential equations using dyson series.
\newblock {\em Quantum}, 8:1369, 2024.

\bibitem{LowSu2024quantumlinearalgorithmoptimal}
Guang~Hao Low and Yuan Su.
\newblock Quantum linear system algorithm with optimal queries to initial state preparation, 2024.

\bibitem{DongLiXue2025}
Dekuan Dong, Yingzhou Li, and Jungong Xue.
\newblock A quantum algorithm for linear autonomous differential equations via padé approximation.
\newblock {\em Quantum}, 9:1770, June 2025.

\bibitem{WuLi2025}
Hsuan-Cheng Wu and Xiantao Li.
\newblock Structure-preserving quantum algorithms for linear and nonlinear hamiltonian systems, 2025.

\bibitem{FangLinTong2023}
Di~Fang, Lin Lin, and Yu~Tong.
\newblock Time-marching based quantum solvers for time-dependent linear differential equations.
\newblock {\em Quantum}, 7:955, March 2023.

\bibitem{JinLiuYu2022}
Shi Jin, Nana Liu, and Yue Yu.
\newblock Quantum simulation of partial differential equations via schr\"odingerization.
\newblock {\em Physical Review Letters}, 133:230602, 12 2024.

\bibitem{AnLiuLin2023}
Dong An, Jin-Peng Liu, and Lin Lin.
\newblock Linear combination of hamiltonian simulation for nonunitary dynamics with optimal state preparation cost.
\newblock {\em Physical Review Letters}, 131(15):150603, October 2023.

\bibitem{AnChildsLin2023}
Dong An, Andrew~M. Childs, and Lin Lin.
\newblock Quantum algorithm for linear non-unitary dynamics with near-optimal dependence on all parameters, 2023.

\bibitem{LowSu2024}
Guang~Hao Low and Yuan Su.
\newblock Quantum eigenvalue processing, 2024.

\bibitem{ShangGuoAnZhao2024}
Zhong-Xia Shang, Naixu Guo, Dong An, and Qi~Zhao.
\newblock Design nearly optimal quantum algorithm for linear differential equations via lindbladians, 2024.

\bibitem{JinLiuMaYu2025}
Shi Jin, Nana Liu, Chuwen Ma, and Yue Yu.
\newblock On the schr\"odingerization method for linear non-unitary dynamics with optimal dependence on matrix queries, 2025.

\bibitem{Li2025}
Xiantao Li.
\newblock From linear differential equations to unitaries: A moment-matching dilation framework with near-optimal quantum algorithms, 2025.

\bibitem{FangGeorgeTong2025}
Di~Fang, David~Lloyd George, and Yu~Tong.
\newblock Qubit-efficient quantum algorithm for linear differential equations, 2025.

\bibitem{BerryChildsCleveEtAl2015}
D.~W. Berry, A.~M. Childs, R.~Cleve, R.~Kothari, and R.D. Somma.
\newblock Simulating {Hamiltonian} dynamics with a truncated {Taylor} series.
\newblock {\em Phys. Rev. Lett.}, 114:090502, 2015.

\bibitem{LowWiebe2019}
G.~H. Low and N.~Wiebe.
\newblock Hamiltonian simulation in the interaction picture.
\newblock {\em arXiv:1805.00675}, 2019.

\bibitem{ChildsWiebe2012}
Andrew~M. Childs and Nathan Wiebe.
\newblock Hamiltonian simulation using linear combinations of unitary operations.
\newblock {\em Quantum Information and Computation}, 12(11--12), November 2012.

\bibitem{ChildsKothariSomma2017}
Andrew~M. Childs, Robin Kothari, and Rolando~D. Somma.
\newblock Quantum algorithm for systems of linear equations with exponentially improved dependence on precision.
\newblock {\em SIAM Journal on Computing}, 46(6):1920–1950, January 2017.

\bibitem{AnLiuWangEtAl2023}
Dong An, Jin-Peng Liu, Daochen Wang, and Qi~Zhao.
\newblock A theory of quantum differential equation solvers: limitations and fast-forwarding, 2023.

\bibitem{AnChildsLinYing2024}
Dong An, Andrew~M. Childs, Lin Lin, and Lexing Ying.
\newblock Laplace transform based quantum eigenvalue transformation via linear combination of hamiltonian simulation, 2024.

\bibitem{LuLiLiuLiu2025}
Rundi Lu, Hao-En Li, Zhengwei Liu, and Jin-Peng Liu.
\newblock Infinite-dimensional extension of the linear combination of hamiltonian simulation: Theorems and applications, 2025.

\bibitem{PocrnicJohnsonKatabarwaWiebe2025}
Matthew Pocrnic, Peter~D. Johnson, Amara Katabarwa, and Nathan Wiebe.
\newblock Constant-factor improvements in quantum algorithms for linear differential equations, 2025.

\bibitem{JinLiuMa2024}
Shi Jin, Nana Liu, and Chuwen Ma.
\newblock Schr\"odingerisation based computationally stable algorithms for ill-posed problems in partial differential equations, 2024.

\bibitem{LowSomma2025}
Guang~Hao Low and Rolando~D. Somma.
\newblock Optimal quantum simulation of linear non-unitary dynamics, 2025.

\end{thebibliography}
\end{document}